\documentclass[conference, 10pt, twocolumn]{ieeeconf}

\IEEEoverridecommandlockouts
\usepackage[usenames]{color}
\usepackage{enumerate}
\usepackage{url}
\usepackage{subfigure}
\usepackage{amsfonts,mathrsfs}
\usepackage{verbatim}
\usepackage{acronym}
\usepackage{pifont}
\usepackage{graphicx}
\usepackage{amsthm}
\usepackage{ifmtarg}
\usepackage{algorithm}
\usepackage{algpseudocode}
\usepackage{mathtools}
\usepackage{cite}
\usepackage{lipsum}
\usepackage{stfloats}

\def\fskip#1{}

\newtheorem{theorem}{Theorem}

\newtheorem{definition}{Definition}

\newtheorem{lemma}{Lemma}

\newtheorem{remark}{Remark}

\renewcommand{\theenumi}{\arabic{enumi}}

\def\1{{\bf 1}}

\newcommand{\remove}[1]{}

\newcommand*{\TitleFont}{%
      \usefont{\encodingdefault}{\rmdefault}{}{n}%
      \fontsize{19}{18}%
      \selectfont}

\makeatletter
\algrenewcommand\ALG@beginalgorithmic{\footnotesize}
\makeatother

\begin{document}
\title{{\TitleFont Managing Price Uncertainty in Prosumer-Centric Energy Trading: A Prospect-Theoretic Stackelberg Game Approach}}
\author{\authorblockN{\large Georges El Rahi, S. Rasoul Etesami, Walid Saad, Narayan Mandayam, and H. Vincent Poor}\vspace{-1cm}\\
\thanks{Georges El Rahi and Walid Saad are with Department of Electrical and Computer Engineering, Virginia Tech, Blacksburg, VA email: (gelrahi,walids)@vt.edu.}
\thanks{S. Rasoul Etesami and is with Department of Industrial Engineering, University of Illinois at Urbana-Champaign, IL (email: etesami1@illinois.edu).}
\thanks{Narayan Mandayam is with Department of Electrical Engineering, Rutgers University,
North Brunswick, NJ (e-mail: narayan@winlab.rutgers.edu).}
\thanks{H. Vincent Poor is with Department of Electrical Engineering, Princeton University,
Princeton, NJ (e-mail: poor@princeton.edu).}
\thanks{This research was supported in part by the NSF under Grants ECCS-1549894, ACI-1541105, ECCS-1549881, CNS-1446621, and ECCS 1549900.}}
\maketitle
\thispagestyle{empty}
\pagestyle{empty}

\begin{abstract}

In this paper, the problem of energy trading between smart grid prosumers, who can simultaneously consume and produce energy, and a grid power company is studied. The problem is formulated as a single-leader, multiple-follower Stackelberg game between the power company and multiple prosumers. In this game, the power company acts as a leader who determines the pricing strategy that maximizes its profits, while the prosumers act as followers who react by choosing the amount of energy to buy or sell so as to optimize their current and future profits. The proposed game accounts for each prosumer's subjective decision when faced with the uncertainty of profits, induced by the random future price. In particular, the  framing effect, from the framework of prospect theory (PT), is used to account for each prosumer's valuation of its gains and losses
with respect to an individual utility reference point. The reference point changes between prosumers and stems from their past
experience and future aspirations of profits. The followers' noncooperative game is shown to admit a unique pure-strategy Nash equilibrium (NE) under  classical game theory (CGT) which is obtained using a fully distributed algorithm. The results are extended to account for the case of PT using algorithmic solutions that can achieve an NE under certain conditions. Simulation results show that the total grid load varies significantly with the prosumers' reference point and their loss-aversion level. In addition, it is shown that the power company's profits considerably decrease when it fails to account for the prosumers' subjective perceptions under PT.  

\end{abstract}

\section {Introduction}
One key enabler of smart grid energy trading and management schemes  is the presence of \emph{prosumers}, i.e., smart grid customers capable of generating and storing their own energy. Indeed, the notable increase in penetration of solar photovoltaic (PV) panels and storage devices at the prosumer side of the grid will lead to novel demand-side energy management (DSM) schemes that will help alleviate the extremely expensive peak consumption hours and match consumption demand to the intermittent renewable 
energy supply of the grid \cite{DRDR}.

A number of recent works \cite{scutari, bitar, hamed} have studied the role of storage devices in grid energy management. For example, the authors in \cite{scutari} propose a storage scheduling and management model, with the goal of  minimizing peak hour consumption. In \cite{hamed}, consumer-based storage scheduling is proposed in an attempt to match energy consumption  to the expected power output of a central wind generation unit. In addition, the work in \cite{bitar} studies the topic of optimal offerings for wind power generators, through the application of a common storage device. Moreover, the work in \cite{chen2012real} adopts a stochastic optimization and robust optimization approach to study demand response for residential appliances under uncertain real-time prices.
  
Game-theoretic methods have been widely applied in the existing DSM literature, given the  coupled interactions between prosumers as discussed in \cite{hamed, hajj, stack, balancing, tamerb, stack2}. For instance, the authors in \cite{hajj} propose a load scheduling technique with a dynamic pricing strategy related to the total consumption of the grid. More particularly, a number of works \cite{stack, balancing, tamerb, stack2} have used the framework of Stackelberg games in order to study the hierarchical interactions between the power company and the grid's consumers. For example,  the authors in \cite{stack} propose a Stackelberg game approach to deal with demand response scheduling under load uncertainty based on real-time pricing  in a residential grid. Similarly, the authors in \cite{balancing} and \cite{stack2} use a
Stackelberg game approach between one power company and multiple users, competing to maximize their profits, with the goal of flattening the aggregate load curve. On the other hand, the authors in \cite{tamerb} propose a Stackelberg game approach between company and consumers, while studying the impact that a malicious attacker could have, through the manipulation of pricing data. In addition, the works in \cite{wei2015energy} and \cite{zugno2013bilevel} have used a Stackelberg game approach to characterize the demand response of consumers with respect to the retail price. In particular, the authors adopted stochastic and robust optimization methods to study energy trading with uncertain market prices. Moreover, a Stackelberg game for energy sharing management of microgrids with photovoltaic prosumers has been proposed in \cite{liu2017energy}.

The main drawback of these works \cite{hamed,hajj,stack,balancing, tamerb,stack2,wei2015energy,zugno2013bilevel,liu2017energy} is the assumption that consumers are fully rational and will thus choose their strategy in accordance to classical game theory (CGT). In practice, as observed by behavioral experimental studies \cite{nobel}, human players can significantly deviate from the rational principles of CGT, when faced with the uncertainty of probabilistic outcomes. In this regard, the framework of \emph{prospect theory} (PT) \cite{nobel}  has been extensively applied to model the irrational behavior of real-life individuals in an uncertain decision making scenario \cite{CPT}.
 In fact, the authors in \cite{Cyumpeng} and \cite{saadt} discuss a storage management framework where the owner can choose to store or sell energy, while accounting for its subjective perceptions, using a PT framework. In \cite{shiftingprospect}, a PT framework is used for DSM to identify optimal customer participation time. However, these works \cite{Cyumpeng, shiftingprospect, walidpf}  do not typically study the conflicting hierarchical interaction between the prosumers and the power company and instead focus on the consumer side of the grid. These works also fail to account for the uncertainty associated with variable or dynamic pricing, which is expected to play a major role in DSM \cite{DRDR}. Even though the combination of Stackelberg games and PT has been applied in other research fields including wireless communication \cite{wir1,wir2}, security games \cite{sec1,sec2}, and transport theory \cite{tran1}, such combination has not been addressed in demand-side energy management problems and from an algorithmic perspective. In particular, these earlier works mainly focus on the weighting effect of PT while here, we consider the framing effect.  
 
\subsection{Contributions}

The main contribution of this paper is a novel hierarchical framework for optimizing energy trading between prosumers and a grid power company, while explicitly taking into account the uncertainty of the future energy price. Our work differs from most of the existing literature on energy trading \cite{scutari, bitar,chen2012real,hamed,hajj,stack,balancing, tamerb,stack2,wei2015energy,zugno2013bilevel,liu2017energy,Cyumpeng,saadt,shiftingprospect,walidpf,wir1,wir2,sec1,sec2,tran1} in several aspects: 1) It models the behavior of prosumers who can both generate and consume energy under price uncertainty and using a Stackelberg game, 2) It provides a simple distributed algorithm with polynomial convergence rate to the unique Stackelberg equilibrium point, and 3) It captures the subjective decision making behavior of the prosumers using \emph{framing} effect in PT. 

In particular, we formulate a single-leader, multiple-follower Stackelberg  game, in which the power company, acts as a leader who declares its pricing strategy in order to maximize its profits, to which prosumers, acting as followers, react by choosing their optimal energy bid. We define a prosumer's utility function that captures the profits resulting from buying/selling energy at the current known price, as well as the uncertain future profits, originating from selling stored energy. In contrast to CGT, we develop a PT framework that models the behavior of prosumers when faced with the uncertainty of future profits. In particular, we account for each prosumer's valuation of its gains and losses, compared to its own individual utility evaluation perspective, as captured via  PT's framing effect \cite{nobel}, by introducing a utility reference point. This reference point represents a prosumer's anticipated profits and originates from previous energy trading transactions and future aspirations of profits, which can differ in between prosumers \cite{CPT}. We show that, under CGT, the followers' noncooperative game admits a unique pure strategy Nash equilibrium (NE). Moreover, under PT, we derive a set of conditions under which the pure strategy NE is shown to exist. In particular, we propose distributed algorithms that allow the prosumers and power company to reach an equilibrium under both CGT and PT. Simulation results show that
the total grid load, under PT, decreases for certain ranges of prosumers' reference points and increases for others, when compared to CGT. The results also highlight the impact of this variation on the  power company's profits, which significantly decrease, when it fails to account for the prosumers' subjective perceptions under PT.

The rest of this paper is organized as follows. Section II presents the system model and formulates the Stackelberg game model. In section III, we present the game solution under CGT and provide a distributed algorithm which can quickly reach the Stackelberg solution of the game.  We extend our results to games under PT in Section IV. Section V presents our simulation results, and finally conclusions are drawn in Section VI.


\section{System model}

Consider the set $\mathcal{N}$ of $N$ grid prosumers. Each prosumer $n \in  \mathcal{N}$  owns an energy storage unit of capacity $Q_{\textrm{max},n}$, and a solar PV panel which produces a daily amount of energy $W_n$. Each prosumer has a known load profile $L_n$ that must be satisfied and an initial stored energy $Q_n$ available in a storage device, originating from an excess of energy at a previous time.

In our model, the power company requires prosumers to declare the amount of energy that they will be buying or selling at the start of the period as done in day-ahead scheduling models used in DSM literature such as in \cite{hamed} and \cite{hajj}. We let $x_n$ be the amount of energy declared by prosumer $n$, where $x_n > 0$ implies an amount of energy that will be bought and $x_n < 0$ will represent the amount of energy that will be sold.  $x_n=0$ indicates that no energy is traded.

The price of selling or buying one unit of energy is related to the total energy  declared by all prosumers. In our pricing model, each prosumer is billed based on the amount that is declared.  We assume that the prosumers are truthful and have no incentive to deviate, given the possible penalties that will be incurred. Next, as done in \cite{shiftingprospect,mohsenian2010optimal}, and \cite{wang2010managing}, we choose a so-called fairness pricing for buying/selling energy which is
proportional to the prosumers' aggregate demand and given by:
\begin{align}
 \rho \left( x_n,\boldsymbol{x}_{-n} \right) = \rho_{\textrm{base}} + \alpha \sum\limits_{n \in \mathcal{N}} x_{n},
\end{align}
where  $\rho_{\textrm{base}}$ and $\alpha$ are design parameters set by the power company. For simplicity, we assume that $\alpha$ is fixed and positive, and that the company only varies $\rho_{\textrm{base}}$ to control the amount of energy bough/sold by the prosumers.\footnote{Note that in this function we can also allow the utility company to adjust $\alpha$. While this does change our analysis from the prosumers' perspective, however, it gives an extra freedom for the utility company to maximize its utility at a cost of solving a more complex optimization problem.} In (1), $\boldsymbol{x}_{-n}$ is a vector that represents the amount of energy declared by all the prosumers in the set  $\mathcal{N}  \setminus \{n\}$. The price of unit of energy $\rho$ is regulated and must be within a range $ [ \rho_{\textrm{min}}, \rho_{\textrm{max}} ]$. Here, we assume that the structure of this pricing function is pre-determined by the utility company and announced a priori to all the prosumers. This function is chosen based on the idea that a higher aggregate demand by the prosumers must naturally increase the energy prices.

The future price of energy  is perceived to be unknown by the prosumers, given the uncertainty related to future solar energy generation and the pricing strategy of the power company. The future price of energy is thus modeled by a random variable $\rho_f$. For simplicity, we assume that $\rho_f$ follows a uniform distribution $[\rho_{\textrm{min}}, \rho_{\textrm{max}} ]$. However, most of our analysis can be extended to the case in which $\rho_f$ follows more general distributions. 


\indent The set of possible values of $x_n$ for each prosumer $n$ is  $\mathcal{X}_n =\left\{   x_n \in \mathbb{R}: x_{n,\min} \leq  x_n \leq x_{n,\max}   \right\}$. $x_{n,\textrm{min}} = -W_n - Q_n +  L_n$ is a prosumer's  maximum sold/minimum bought energy.
$x_{n,\textrm{max}} = -W_n - Q_n +  L_n + Q_{\textrm{max},n}$, is the maximum energy that prosumer $n$ can purchase. For a chosen energy bid $x_n$, the prosumer's utility function will be:
\begin{align}\label{eq:1} 
U_{n} \left( {x_n,\boldsymbol{x}_{-n}}, \rho_{\textrm{base}}  \right) &=-\Big( \rho_{\textrm{base}} + \alpha ( x_n + \sum\limits_{m \in \mathcal{N} \setminus n} x_{m} ) \Big) x_n\cr
&\qquad+\left( W_n + Q_n -L_n + x_n \right) \rho_f.  
\end{align}
In \eqref{eq:1}, the first term represents the revenue/cost of prosumer $n$ at the current time, while the second term represents the future monetary value associated with unsold energy. In particular, $W_n + Q_n - L_n + x_n$ is the amount of energy that prosumer $n$ will have in its storage in the future. The prosumers' actions are coupled through the energy price and they will thus be competing to maximize their respective revenues. On the other hand, the power company will purchase (sell) the energy bough (sold) by the prosumers in the energy market at the current market clearing price $\rho_{\textrm{mar}}$. Given the current market price, the power company's utility function is given by:

\vspace{-0.3cm}
\begin{small}
\begin{align}
U_{pc}(\boldsymbol{x}, \rho_{\textrm{base}} ) =  \left( \rho_{\textrm{base}}  + \alpha  \sum\limits_{n \in \mathcal{N} } x_{n} \right) \sum\limits_{n \in \mathcal{N} } x_{n} - \rho_{\textrm{mar}}  \sum\limits_{n \in \mathcal{N} } x_{n},
\end{align}
\end{small}where the first term represents the revenue that the utility company earns by selling (buying) $\sum_n x_n$ energy units to prosumers at the price of $ \rho_{\textrm{base}}  + \alpha  \sum_{n} x_{n}$, while the second term is the cost of purchasing $\sum_n x_n$ energy units at the clearing price of $\rho_{\textrm{mar}}$ from the energy market.
  
The power company's revenues are clearly affected by the prosumers and their energy bids. On the other hand, since the prosumers react to the power company's choice of $\rho_{\textrm{base}}$, the prosumers' utility is directly affected by the power company's action. We thus model the energy trading problem as a hierarchical Stackelberg game \cite{walid} with the power company acting as leader, and the prosumers acting as followers.


%
\subsection{Stackelberg game formulation}
We formulate a  single-leader, multiple-follower  Stackelberg game \cite{walid}, between the power company and the prosumers. The power company (leader), will act first by choosing $\rho_{\textrm{base}}$ to maximize its profits. The prosumers, having received the power company's pricing strategy, will engage in a noncooperative game. In fact, the final price of energy is proportional to the grid's total load, to which each prosumer contributes. We first formulate the prosumers' problem under CGT as follows:
\begin{flalign}
 \underset{x_n}{\textrm{max}} \,\,  U^{\textrm{CGT}}_n(\boldsymbol{x},\rho_{\textrm{base}}): =\, & \mathbb{E}_{\rho_{f}}[U_n\left( {x_n,\boldsymbol{x}_{-n}}, \rho_{\textrm{base}}  \right)] \\
							   & \textrm{s.t}   \,\,x_n \in \left[ x_{\textrm{min},n}, x_{\textrm{max},n}   \right].    \nonumber                        
\end{flalign}

In (4), prosumer $n$ attempts to maximize its expected profits, given the actions of other prosumers and the power company. The previous formulation assumes all prosumers to be rational expected utility maximizers. Moreover, the power company's problem will be:
\begin{align}
 &\underset{\rho_{\textrm{base}}} {\textrm{max}} \,\,  U_{pc}(\boldsymbol{x}, \rho_{\textrm{base}}), \,\,\,  
							   \textrm{s.t}   \,\, \rho_{\textrm{base}} \in \left[ \rho_{\textrm{min}}, \rho_{\textrm{max}}   \right].                         
\end{align}
To solve this game, one suitable concept is that of a Stackelberg equilibrium (SE) as the  game-theoretic solution of our model.
  
\begin{definition}\label{def:stackelberg}
A strategy profile ($\boldsymbol{x}^*, \rho_{\emph{base}}^*$) is  a \emph{Stackelberg equilibrium} if it satisfies the following conditions:
\begin{flalign} \label{eq:6}
\nonumber U^{\rm CGT}_{n} ({x}_{n}^* ,\boldsymbol{x}_{-n}^*, \rho_{\emph{base}}^*) &\geq   U^{\rm CGT}_{n}({x}_{n},\boldsymbol{x}_{-n}^*, \rho_{\emph{base}}^*) \ \ \ \forall n \in \mathcal{N}, \\
\underset{\boldsymbol{x}^*} {\emph{min}} \, U_{pc} (\boldsymbol{x}^*, \rho_{\emph{base}}^*) &=  \underset{\rho_{\emph{base}}} {\emph{max}}\, \underset{\boldsymbol{x}^*} {\emph{min}} \, U_{pc} (\boldsymbol{x}^*, \rho_{\emph{base}}),
\end{flalign}
where $\boldsymbol{x}_n^*$ is the solution to problem $(4)$ for all prosumers in $\mathcal{N}$,  and  $\rho_{\emph{base}}^*$  is the solution to problem $(5)$.
\end{definition}

\begin{remark}
Note that in Definition \ref{def:stackelberg}, in the case where the followers' problem admits a unique solution $\boldsymbol{x}^*$, the second condition in (\ref{eq:6}) reduces to $U_{\textrm{pc}} (\boldsymbol{x}^*, \rho_{\emph{base}}^*) = \max_{\rho_{\rm base}} U_{\textrm{pc}} (\boldsymbol{x}^*, \rho_{\emph{base}})$. 
\end{remark}

It is worth noting that our Stackelberg formulation is based on a \emph{static} game. However, our solution to this problem is based on the notion of a repeated game approach in which the prosumers frequently interact with the utility company in order to find their equilibrium strategies. This is practically important as it is a step toward analyzing more complex senarios in which the smart grid's environment dynamically changes from one time instant to the other. For instance, one can consider a multi-stage game where for each prosumer $n$, $W_n$ and $L_n$ dynamically vary with time $t$, while $Q_n$ is affected by $W_n, L_n$, and action $x_n$ taken at previous and current times. As such, to capture this dynamic nature, our proposed static game model can be expanded to a dynamic \emph{stochastic game} \cite{fink1964equilibrium,shapley1953stochastic,etesami2016stochastic} with transition equations describing the evolution of the states, corresponding to $W_n(t)$, $L_n(t)$, and $Q_n(t)$, with respect to time depending on the control inputs $x(t):=[x_1(t),...,x_n(t)]$ and previous states. In this case, the state of the game at time $t$ consists of $W_n(t)$, $L_n(t)$, and $Q_n(t)$ for prosumer $n$ at stage $t$. The chosen optimal action by each player at time $t$, i.e. control inputs, would depend on the state at which the game is. In this respect, our static game analysis here provides analytical and algorithmic solution approaches through which optimal strategies, for the prosumers and energy company, can be obtained at the \emph{stationary} states of the stochastic game.\footnote{In a stochastic game framework, a stationary strategy consists of obtaining the optimal strategy for a certain state regardless of the history of the game or the particular time instant at which the game is played.} In other words, our repeated single stage game analysis can be viewed as a solution to the multi-stage stochastic game which has reached its stationary condition (i.e., at the stationary state, it appears as if one is repeatedly playing the same stationary game). In particular, the optimal stationary strategies can be extracted from our static game analysis under the stationarity condition.

\section{Game Solution under CGT}

The analysis under CGT assumes that all prosumers are expected utility maximizers. Thus, we seek to find a solution that solves both problems $(4)$ and $(5)$, while satisfying (\ref{eq:6}). First, we start by solving the follower's problem while assuming the the leader's action is fixed to $\rho_{\textrm{base}}$. We now introduce the following notations:
\begin{flalign}
&\theta:=\rho_{\textrm{base}}-\frac{\rho_{\max}+\rho_{\min}}{2},\ \ \  \bar{x}_{-n}:=\sum_{k\neq n}x_k, \nonumber \\
&\delta_n:=(W_n+Q_n-L_n)\frac{\rho_{\max}+\rho_{\min}}{2}, \ \ n\in\mathcal{N}.\nonumber
\end{flalign}
Here, the expected utility of prosumer $n\in\mathcal{N}$ will be:
\begin{align}\label{eq:utility-modified}
U^{\textrm{CGT}}_{n}(x_n,\bar{x}_{-n},\rho_{\rm base})=-\alpha x^2_n-(\theta+\alpha\bar{x}_{-n})x_n+\delta_n.
\end{align}

Next, we denote by $x^{r}_{n}$ the best response of player $n$, which is the solution of problem $(4)$, given that all the other players choose a specific strategy profile $\boldsymbol{x}_{-n}$. The following theorem explicitly characterizes the best response of each prosumer $n$.

\begin{theorem}\label{thm:best-response}
The best response of player $n$ is given by:
\begin{flalign}\label{eq:best-response}
x^{r}_{n}(\bar{x}_{-n})&= \begin{cases} -\frac{\theta}{2\alpha}-\frac{\bar{x}_{-n}}{2} & \mbox{if } -\frac{\theta}{2\alpha}-\frac{\bar{x}_{-n}}{2}\in[x_{n,\min},x_{n,\max}], \\
x_{n,\min} & \mbox{if } -\frac{\theta}{2\alpha}-\frac{\bar{x}_{-n}}{2}\leq x_{n,\min},\\
 x_{n,\max} & \mbox{\emph{else}.}\end{cases}
\end{flalign}  
\end{theorem}
\begin{proof}
See Appendix \ref{apx:thm_best-response}.
\end{proof}

In fact, one can rewrite the  best responses of all the players in Theorem \ref{thm:best-response} in a combined single matrix form. We define $\boldsymbol{A}$ to be an $n\times n $ matrix with all entries equal to -$\frac{1}{2}$ except the diagonal entries which are $0$, i.e., $A_{ij}=-\frac{1}{2}$ if $j\neq i$, and $A_{ij}=0$, otherwise. Let $\boldsymbol{a}=-\frac{\theta}{2\alpha}\boldsymbol{1}$ where $\boldsymbol{1}$ is the vector of all $1$'s. Then, we can rewrite  \eqref{eq:best-response} for all players as
\begin{align}\label{eq:recursion}
\boldsymbol{x}^{r}=\Pi_{\Omega}\big[\boldsymbol{a}+\boldsymbol{A}\boldsymbol{x}\big],
\end{align}
where $\Pi_{\Omega}[\cdot]$ is the projection operator on the $n$ dimensional cube  $\Omega:=\prod\limits_{n \in \mathcal{N}}[x_{n,\min},x_{n,\max}]$ in $\mathbb{R}^n$. Our analysis will later use this closed-form representation of the best response dynamics.          

\subsection{Existence and uniqueness of the followers' NE under CGT}

One key question with regard to the prosumers' game is whether such a game admits a pure-strategy NE. This is important as it allows us to stabilize the demand market in an equilibrium where each prosumer is satisfied with its payoff, as shown next. 

\begin{theorem}\label{thm:existence-uniqueness}
The prosumers' game admits a unique pure-strategy NE.
\end{theorem}
\begin{proof}
See Appendix \ref{apx:thm_existence-uniqueness}.
\end{proof}

\subsection{Distributed learning of the followers NE}

\begin{algorithm}[t!]
\caption{The relaxation learning algorithm}\label{alg:relaxing}
\begin{algorithmic}[0]
\State Given that at time step $t=1,2,\ldots$ players have requested $(x_1(t),\ldots,x_n(t))$ units of energy, at the next time step player $n\in \mathcal{N}$ requests $x_n(t+1)$ energy units given by  
\begin{align}\nonumber
x_n(t+1)=(1-\frac{1}{\sqrt{t}})x_n(t)+\frac{1}{\sqrt{t}}x^{r}_n(t),
\end{align}
where $x^{r}_n(t)$ denotes the best response of player $n$, given the actions of all other players $\boldsymbol{x}_{-n}(t)$ at time step $t$.
\end{algorithmic}
\end{algorithm}Next, we propose a distributed learning algorithm which converges in a polynomial rate to the unique  pure-strategy NE of the prosumers' game as formally stated in Algorithm \ref{alg:relaxing}. At each stage of Algorithm \ref{alg:relaxing}, prosumer $n$ selects its next action as a convex combination of its current action and its best response at that stage. One of the main advantages of Algorithm \ref{alg:relaxing} is that it can be implemented in a completely distributed manner as each prosumer needs only to know its own actions and best response function, and does not require any information about others' actions. Moreover, the prosumers do not need to keep track of their actions history which is the case in many other learning algorithms. Note that Algorithm \ref{alg:relaxing} can be viewed as a special case of more general algorithms known as \textit{relaxation} algorithms \cite{krawczyk2000relaxation}.

The idea behind Algorithm \ref{alg:relaxing} is that each player initially puts more weight on its best response in order to explore faster other possible actions with better payoffs. As the time elapses, the prosumers' actions become closer to their optimal actions and, hence, the prosumers exploit their current actions by putting more weight on their own actions. While the exploration coefficient $\frac{1}{\sqrt{t}}$ can be replaced by other possible coefficients, we have chosen $\frac{1}{\sqrt{t}}$ to optimize the speed of convergence. Finally, note that the implementation of Algorithm \ref{alg:relaxing} is made possible by a bidirectional communication between the power company and the prosumers, provided by the smart meters. In fact, at each iteration, the prosumer would send the power company its current strategy, and would receive the updated energy price.

Next, we consider the following definition which will be handy in proving our main convergence result in this section: 
\begin{definition}\label{def:Nikaido}
Given an $n$ players game with utility functions $\{u_n(\cdot)\}_{n \in \mathcal{N}}$ and any two action profiles $\boldsymbol{x}$ and $\boldsymbol{y}$, the \emph{Nikaido-Isoda} function associated with this game is given by $\Psi(\boldsymbol{x},\boldsymbol{y}):= \sum_{n \in \mathcal{N}}[u_n(y_n,x_{-n})-u_n(x_n,x_{-n})]$.
\end{definition}

The Nikaido-Isoda function measures the social income due to selfish deviation of  individuals. This function admits several key properties. As an example we always have $\Psi(\boldsymbol{x},\boldsymbol{x})=0, \forall \boldsymbol{x}$. Moreover, given a fixed action profile $\boldsymbol{x}$, $\Psi(\boldsymbol{x},\boldsymbol{y})$ is maximized when $y_n,$ equals the best response of player $n$ with respect to $\boldsymbol{x}_{-n}$. In particular, for such a best response action profile $\boldsymbol{y}$, $\Psi(\boldsymbol{x},\boldsymbol{y})=0$ if and only if $\boldsymbol{x}$ is a pure strategy NE of the game. While the Nikaido-Isoda function has been used earlier to prove convergence of certain dynamics to their equilibrium points \cite{facchinei2007generalized,krawczyk2000relaxation}, however it usually fails to provide an explicit convergence rate. In the following theorem we leverage the Nikaido-Isoda function associated to the prosumers' game to measure the distance of outputs of Algorithm \ref{alg:relaxing} from the Nash equilibrium, and hence obtain an explicit bound on the convergence rate of this algorithm.
\begin{theorem}\label{thm:convergence-rate}
If every prosumer updates its energy request bid based on Algorithm \ref{alg:relaxing}, then their action profiles will jointly converge to an pure strategy NE. After $t$ steps the joint actions will be an $\epsilon$-NE where $\epsilon=O(t^{-\frac{1}{4}})$ (i.e., the convergence rate to an NE is $O(t^{-\frac{1}{4}})$).
\end{theorem}
\begin{proof}
See Appendix \ref{apx:thm:convergence-rate}.
\end{proof}   

As it has been shown in Appendix \ref{apx:thm_existence-uniqueness}, the prosumers' game is a \emph{concave game} \cite{rosen1965existence}, which is known to admit a distributed learning algorithm for obtaining its NE points (see, e.g., \cite[Theorem 10]{rosen1965existence}). However, in general obtaining distributed learning algorithms with provably fast convergence rates to NE points in concave games is a challenging task. Therefore, one of the main advantages of Theorem \ref{thm:convergence-rate} is that it establishes a polynomial convergence rate for the relaxation Algorithm \ref{alg:relaxing} leveraging rich structure of the prosumers' utility functions.

\begin{remark}
In fact, one of the advantages of our formulation compared to similar models such as \cite{zugno2013bilevel} is its computational tractability as it admits polynomial time distributed algorithms for finding its equilibrium points, regardless of the number of players in the game (Theorem \ref{thm:convergence-rate}).  
\end{remark}

\subsection{Finding the Stackelberg Nash equilibrium under CGT}\label{sec:methods-SE}

While Algorithm \ref{alg:relaxing} achieves a unique pure strategy for the prosumers' game under CGT, our final goal is to obtain the Stackelberg  equilibrium of the entire game. For this purpose, we leverage Algorithm \ref{alg:relaxing} using one of the following methods to construct the SE of the entire market under CGT: 
\subsubsection{Method 1}
The Stackelberg equilibrium of the game can be found by solving the following non-linear optimization problem.
Let $\boldsymbol{x}^*(\rho_{\textrm{base}})$ be the unique NE obtained by the followers when the power company's action is $\rho_{\textrm{base}}$. Note that $\boldsymbol{x}^*(\rho_{\textrm{base}})$ is a well-defined continuous function of $\rho_{\textrm{base}}$. First, the power company solves the following optimization problem a priori to find its unique optimal action $\rho^*_{\textrm{base}}$ and announces it to the prosumers. The problem is defined as:   
\begin{align}\label{eq:non-linear}
&\max_{\rho_{\textrm{base}}} U_{\textrm{pc}}(\boldsymbol{x}^*(\rho_{\textrm{base}}),\rho_{\textrm{base}}) \cr 
&s.t. \ \ \ \boldsymbol{x}^*(\rho_{\textrm{base}})=\Pi_{\Omega}[\boldsymbol{a}+\boldsymbol{A}\boldsymbol{x}^*(\rho_{\textrm{base}})].
\end{align}

In fact, one can characterize the unique pure-strategy NE of the prosumers' game given in \eqref{eq:non-linear} in more detail. Since at equilibrium every player must play its best response in \eqref{eq:recursion}, therefore, an action profile $\boldsymbol{x}^*$ is an equilibrium if and only if we have $\boldsymbol{x}^*=\Pi_{\Omega}[\boldsymbol{a}+\boldsymbol{A}\boldsymbol{x}^*]$, which means $\boldsymbol{x}^*=\textrm{argmin} \|\boldsymbol{z}-(\boldsymbol{a}+\boldsymbol{A}\boldsymbol{x}^*) \| ^2,  \boldsymbol{z}\in \Omega$ . Since the former is a convex optimization problem, we can write its dual as
\begin{align}\nonumber
&\max \ \mathcal{D}(\boldsymbol{\mu},\boldsymbol{\nu}):=-\frac{1}{4}\|\boldsymbol{\mu}-\boldsymbol{\nu}\|^2\!+\!(\boldsymbol{\mu}-\boldsymbol{\nu})'(\boldsymbol{a}+\boldsymbol{A}\boldsymbol{x}^*)\!-\!\boldsymbol{\mu}'{\bf 1} \cr 
&\hspace{3cm} \boldsymbol{\mu},\boldsymbol{\nu}\ge 0,
\end{align}
where $\boldsymbol{\mu}=(\mu_1,\ldots,\mu_n)$ and $\boldsymbol{\nu}=(\nu_1,\ldots,\nu_n)$ are the dual variables corresponding to the constraints $z_n\leq x_{n,\max}$ and $z_n\ge x_{n,\min}$, respectively. We denote the optimal solution of the dual by $(\boldsymbol{\mu}^*,\boldsymbol{\nu}^*)$. Since, we already know that $\boldsymbol{x}^*$ is the optimal solution of the primal, due to the strong duality the values of the primal and dual must be the same, i.e., $\mathcal{D}(\boldsymbol{\mu}^*,\boldsymbol{\nu}^*)= \|(I-\boldsymbol{A})\boldsymbol{x}^*-\boldsymbol{a}\|^2$. Moreover, the KKT conditions must hold at the optimal solution \cite{bertsekas2003convex}, which, together with $\mathcal{D}(\boldsymbol{\mu}^*,\boldsymbol{\nu}^*)= \|(I-\boldsymbol{A})\boldsymbol{x}^*-\boldsymbol{a}\|^2$ provide the following system of $3n$ equations with $3n$ variables $(\boldsymbol{x}^*,\boldsymbol{\mu}^*,\boldsymbol{\nu}^*)$ which characterizes the equilibrium point $\boldsymbol{x}_n^*$ using dual variables:
\begin{align}\label{eq:dual-characterization} 
&\mathcal{D}(\boldsymbol{\mu}^*,\boldsymbol{\nu}^*)= \|(I-\boldsymbol{A})\boldsymbol{x}^*-\boldsymbol{a}\|^2, \cr \nonumber
& \mu^*_n(x^*_n- x_{\max,n})=0, \\ 
& \nu^*_n(x^*_n- x_{\min,n})=0, \ \ \forall n\in \mathcal{N}.
\end{align} 
Solving these equations can be used to derive the unique pure-strategy NE of the followers' game.

We next present a second method, which does not require the power company to solve the non-linear inequalities in \eqref{eq:dual-characterization}. In addition, the second method allows the players to reach the SE quickly and efficiently in $1/\epsilon^5$ steps and will be mainly used in our simulation results in Section V. 
\subsubsection{Method 2}
Given any small $\epsilon>0$ for which the power company and the prosumers want to find their $\epsilon$-SE with precision $\epsilon$ (i.e., no one can gain more than $\epsilon$ by deviating), the company partitions its action interval and sequentially announces prices
 $\rho_{\textrm{base}}=k \epsilon, k=1,\ldots, \lfloor \frac{1}{\epsilon} \rfloor$. For each 
such price $\rho_{\textrm{base}}$, prosumers obtain their $\epsilon$-equilibrium in no more than $\frac{1}{\epsilon^4}$ steps, and the company must  repeat this process at most $\frac{1}{\epsilon}$ steps and choose the action that maximized its utility. The running time in this case will be $\frac{1}{\epsilon^5}$ to find an $\epsilon$-SE.


Our analysis thus far assumed that all prosumers are fully rational and their behavior can thus be modeled using CGT. However, this assumption might not hold , given that prosumer are humans that can have different subjective valuations on their uncertain energy trading payoffs.  Next, we extend our result using PT \cite{nobel} to model the behavior of prosumers when faced with the unknown future price of energy and thus the actual value of the stored energy.  

\begin{figure*}[t]
  \begin{center}
\begin{equation}\label{eq:PT} 
{\scriptsize U^{\textrm{PT}}_{n}(x_n,\bar{x}_{-n},\rho_{\textrm{base}}) =\begin{cases}
  \dfrac{(c \rho_{\textrm{max}}+ d - R_n)^{\beta^{+}+1} - (c \rho_{\textrm{min}}+ d - R_n)^{\beta^{+} +1}}{c(\beta^{+} + 1)\rho_\textrm{d}},  &\textrm{ if }   R_n < \rho_{\textrm{min}} c + d,   \\
   \dfrac{(c \rho_{\textrm{max}}+ d - R_n)^{\beta^{+} + 1}}{c(\beta^{+}+1)\rho_\textrm{d}} -\lambda_n\dfrac{(-c \rho_{\textrm{min}}- d + R_n)^{\beta^{-} +1}}{c(\beta^{-}+1)\rho_\textrm{d}},  &\textrm{ if }  \rho_{\textrm{min}} c + d< R_n < \rho_{\textrm{max}} c + d,   \\
  \dfrac{\lambda_n(-c\rho_{\textrm{max}}- d + R_n)^{\beta^{-} + 1} -\lambda_n (-c \rho_{\textrm{min}}- d + R_n)^{\beta^{-} + 1}}{c(\beta^{-}+1)\rho_\textrm{d}},  &\textrm{ if }   \rho_{\textrm{max}} c + d < R_n.   \\
\end{cases}
\vspace{-0.4cm} }
\end{equation}
\end{center}\vspace{-0.4cm}
\end{figure*}

\section{Prospect  theoretic analysis}


In a classical game-theoretic framework, an individual evaluates an objective expected utility. However, in a real-life setting, empirical studies \cite{nobel}, have shown that decision makers, tend to deviate noticeably from the rationality axioms, when subjected to uncertain payoffs. The most prominent of such studies was that done by Kahneman and Tversky within the context of prospect theory \cite{nobel}, which won the 2002 Nobel prize in economic sciences.

The utility framing notion is one of the two main tenets of prospect theory. As observed in real-life experimental studies, utility framing states that each individual perceives a utility as either a loss or a gain, after comparing it to its individual reference point \cite{nobel}. The reference point is typically different for each individual and originates from its past experiences and future aspirations of profits. Furthermore, individuals tend to evaluate losses in a very different manner compared to gains. The main axioms of utility framing are summarized as follows:

\begin{itemize}
  \item Individuals perceive utility according to changes in value with respect to a reference point rather than an absolute value.
  \item Individuals assign a higher value to  differences between small gains or losses close to the reference point in comparison to those further away. Te effect is referred to as diminishing sensitivity, and is captured by the coefficients ${\beta^+}$ and ${\beta^-}$.
  \item Individuals feel greater aggravation for losing a sum of money than  satisfaction associated with gaining the same amount of money. This phenomenon is referred to as loss aversion and is captured by the aversion coefficient $\lambda$.
 \end{itemize}
 It is worth noting that PT differs from other risk measures such as \emph{Conditional Value at Risk} (CVaR) \cite{CVaR} which evaluates the market risk based on the \emph{expected} value of the risk at some future time. The underlying assumption in evaluating CVaR is that the risk is measured based on the conventional expectation of the future uncertain price, while in PT, the expectation is replace by subjective perception of the individuals which up to some extent introduces a notion of bounded rationality into the model.

 \subsection {Energy Trading Analysis through Utility Framing}
 
 In our model, a prosumer's uncertainty originates from the unknown future energy price and power company pricing strategy. Consequently, we will analyze the effect of the key notion of  \emph{utility framing} from PT. Utility framing states that a utility is considered a gain if it is larger than the reference point, while it is perceived as a loss if it is smaller than that reference point. This reference captures a prosumer's anticipated profits and  originates from past energy trading transactions and future aspirations of profits, which can differ in between different prosumers \cite{CPT}. Let $R_n$ be the reference point of a given prosumer $n$. Thus, to capture such subjective perceptions, we use PT framing \cite{CPT} to redefine the utility function:\vspace{-0.5cm} 
 
\begin{footnotesize}
\begin{align}\label{eq:rule}
V\left(U_n\left(\boldsymbol{x}, \rho_{\textrm{base}}\right) \right)\!=\!\begin{cases}
\left( U_n(\boldsymbol{x}, \rho_{\textrm{base}})\!-\!R_n\right)^{\beta^+} &\!\!\!\textrm{ if } U_n(\boldsymbol{x}, \rho_{\textrm{base}})\!>\! R_n,\nonumber \\
-\lambda_n  \left(R_n\!-\!U_n(\boldsymbol{x}, \rho_{\textrm{base}}\right)^{\beta^-} &\!\!\!\textrm { if } U_n(\boldsymbol{x}, \rho_{\textrm{base}})\!<\!R_n,
\end{cases}\\
\end{align}\end{footnotesize}where $\beta^-,\beta^+\in(0,1]$ and $\lambda \geq 1$. $V(\cdot)$ is a framing value function, concave in gains and convex in losses with a larger slope for losses than for gains \cite{CPT}. The expected utility function of prosumer $n$ under PT, for a given action profile $\boldsymbol{x}$, is given by (\ref{eq:PT}) where $c:= W_n + Q_n + x_n - L_n$, $d:= -\left( \rho_{\textrm{base}} + \alpha (x_n +   \bar{x}_{-n}) \right)x_n$, and $\rho_{\textrm{d}}:=\rho_{\textrm{max}} - \rho_{\textrm{min}} $.

\subsection{Existence and uniqueness of the NE under PT}

To study the existence of the followers' NE under PT, we analyze the concavity of the utility function in (\ref{eq:PT}). The concavity of the PT utility function provides a sufficient condition to conclude the existence of at least one pure-strategy NE \cite[Theorem 1]{rosen1965existence}. Here, we note that prosumer $n$'s expected utility function can take multiple forms over the product action space $\Omega$, depending on the conditions in (\ref{eq:PT}). It is thus challenging to prove that the utility function is concave, which makes it extremely difficult to analyze the existence and uniqueness of the followers' NE. Thus, we inspect a number of conditions under which the PT utility function is concave. Here, for simplicity and to provide more closed-form solutions, we disregard the diminishing sensitivity effect and thus set $\beta^{+}=\beta^{-}=1$. The following theorem provides  sufficient conditions under which the prosumers' game under PT admits a pure NE.

\begin{theorem}\label{thm:PT_Cases}
In either of the following cases, the prosumers' game under PT admits at least one pure strategy NE:
\begin{itemize}
\item { Case 1}: $\Delta_1 >0, \emph{and} \,\,  x_{r1} < x_{n,\emph{min}}, x_{n,\emph{max}} < x_{r2}$.
\item { Case 2}: $\Delta_2 < 0,  \emph{or} \,\,  x_{n,\emph{max}} < x_{r3}, \emph{or} \,\,\,\, x_{r4} < x_{n,\emph{min}}$.
\item { Case 3}: $(\Delta_2 >0, x_{r3} < x_{n,\emph{min}}, x_{n,\emph{max}} < x_{r4})$, \emph{and} $(\Delta_1 < 0,  \emph{or} \,\,  x_{n,\emph{max}} < x_{r1}$ \emph{or} $x_{r2} < x_{n,\emph{min}})$, \emph{and} $\left( \tiny{ x_{\emph{max},n} < 1-\frac {b_1}{a_1}} \right)$, 
\end{itemize}
where
{\small
\begin{align}\nonumber
&k_n:= W_n + Q_n - L_n,\cr 
&\Delta_1:= (\rho_{\emph{min}} - \rho_{\emph{base}} - \alpha\bar{x}_{-n})^2 + 4\alpha (k_n \rho_{\emph{min}} - R_n),\cr 
&\Delta_2:= (\rho_{\emph{max}} - \rho_{\emph{base}} - \alpha\bar{x}_{-n})^2 + 4\alpha (k_n \rho_{\emph{max}} - R_n),\cr 
&x_{r1,r2}:= \frac{\pm \sqrt{\Delta_1} +  (\rho_{\emph{min}} - \rho_{\emph{base}}  - \alpha\bar{x}_{-n}) }{2\alpha},\cr 
&x_{r3,r4}:= \frac{\pm\sqrt{\Delta_2} +  (\rho_{\emph{max}} - \rho_{\emph{base}}  - \alpha\bar{x}_{-n})}{2\alpha},\cr 
&m_1:= 64(W_n + Q_n - L_n), \ a_1= 48 \alpha^2 (1-\lambda_n),\cr 
&b:= (176\alpha^2k_n + 32\alpha(\rho_{\emph{base}} - \rho_{\emph{max}} + \alpha\bar{x}_{-n}))(1-\lambda_n).
\end{align} }   
\end{theorem}
\begin{proof}
See Appendix \ref{apx:thm:PT_Cases}.
\end{proof}

As an immediate corollary of Theorem \ref{thm:PT_Cases}, under any of the above conditions, one can again obtain the SE of the entire market for PT prosumers using the same procedure used under CGT (i.e., using Algorithm \ref{alg:relaxing} in the prosumers' side together with either of the methods in Subsection \ref{sec:methods-SE}). This is simply because, under any of the conditions in Theorem \ref{thm:PT_Cases}, the prosumers' game again becomes a concave game which is sufficient for the convergence of Algorithm \ref{alg:relaxing}. It is worth noting that, in general, using PT rather than CGT will change the results pertaining to the existence of an NE (see e.g.,
\cite{metzger2010equilibria}). However, Theorem \ref{thm:PT_Cases} provides a sufficient condition under which the same existence results derived for CGT still hold under PT.

Finally, whenever the concavity of the game cannot be guaranteed, we propose a sequential best response algorithm, that build on our previous work in \cite{rahi}. This is a special case of Algorithm $1$, where $x_n(t+1) = x^{r}_n(t)$, and where players update their strategy sequentially instead of simultaneously. An analytical proof of existence/convergence is challenging, given that no proof for the game's concavity could be derived, as previously discussed. However, when it converges, this algorithm is guaranteed to reach an NE. In fact, as observed from our simulations in Section V, the algorithm always converged and found a pure-strategy NE, for all simulated scenarios.

\section{Simulation Results and analysis}\label{sec:simulations}

For our simulations,  we consider a smart grid with $N=9$ prosumers, unless otherwise stated, each of which having a load $L_n$ arbitrarily chosen within the range $[10, 30]$ kWh. In addition, the storage capacity $Q_{\textrm{max},n}$ is set to 25 kWh and $\alpha = 1/N$. $\beta^+$ and $\beta^-$ are taken to be both equal to 0.88 and $\lambda=2.25$, unless stated otherwise  \cite{CPT}. We set, $\rho_{\textrm{base}} = \$ 0.04$ and $R_n = \$ 1$, unless stated otherwise. When the leader's action is not fixed, method $2$ from Section III-C was used to find the SE.   
 \begin{figure}[t]\label{fig:2}
  \begin{center}
 \includegraphics[width=0.9\linewidth,height=.16\textheight]{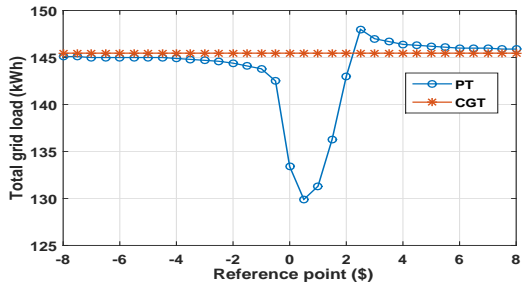}
  \caption{\label{fig:2}Total grid load under expected utility theory and prospect theory.}
  \end{center}\vspace{-0.3cm}
\end{figure}

\indent Fig. \ref{fig:2} compares the effect of different prosumer reference points on the total  energy sold or bought for both CGT and PT, while fixing the power company's action. For CGT, a prosumer's reference point is naturally irrelevant. For the PT case, for a reference point below $-\$2$, the prosumers' action profile is not significantly affected compared to CGT, since most potential payoffs of the action profile are still viewed as gains, above the reference point. As the reference point increases from $-\$2$ to $\$0.5$, the total energy consumed will decrease from around $145$ kWh to $130$ kWh, since some of the potential payoffs of the current action profile will start to be perceived as losses, as they cross the reference point. Given that losses have a larger weight under PT compared to CGT, the expected utility of the current strategy profile will significantly decrease thus causing the followers to exhibit a risk-averse behavior. In fact, as  some of the potential future profits are perceived as losses, a prosumer will sell more energy at the current time slot. As the reference point increases from $\$0.5$ to $\$2$, the present profits are now perceived as losses, and prosumers will start exhibiting risk-seeking behavior. In fact, each prosumer will consider the present profit as insignificant and will thus store more energy in the hope of selling it in the future at higher prices. Finally, as the reference point approaches \$8, the effect of uncertainty will gradually decrease, given that all profits are now perceived as losses.  
We note that even a small difference in perception ($\$1.5$) caused the total grid load to shift from $145$ kWh to $130$ kWh. This highlights the importance of behavioral analysis and prosumer subjectivity when assessing the performance of dynamic pricing strategies.
\begin{figure}[t]\label{fig:3}
  \begin{center}\vspace{-0.05cm}
\includegraphics[width=0.95\linewidth,height=.16\textheight]{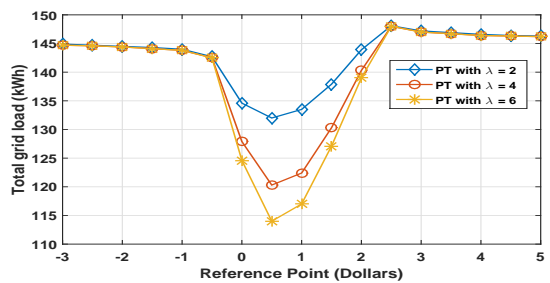}
  \caption{Effect of varying the loss multiplier $\lambda$.}\label{fig:3}
  \end{center}\vspace{-0.35cm}
\end{figure}

\indent Fig. \ref{fig:3} shows the effect of the loss multiplier $\lambda$ on the total energy purchased, for a fixed power company strategy. The loss multiplier maps the loss aversion of prosumers when assessing their utility outcomes. The effect of framing is more prominent as the loss multiplier increases. For instance, the prosumers will exhibit more risk averse behavior for a reference point in the range of $[-\$0.5, \$2.5]$. As seen from Fig. \ref{fig:3}, as $\lambda$ increase from $2$ to $6$, the total load would decrease by up to $14 \%$. In fact, to avoid the large losses, the prosumers will decrease the energy they purchase at the current risk free energy price. 
\begin{figure}[t]\label{fig:4}
  \begin{center}
   \hspace{-0.5cm}\includegraphics[width=0.89\linewidth,height=0.21\textheight]{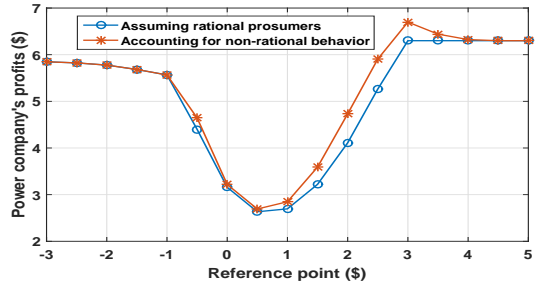}
  \vspace{-1cm}\caption{\label{fig:4}Power company's profits}
  \end{center}\vspace{-0.4cm}
\end{figure}

\begin{figure}[t]\label{fig:5}
  \begin{center} \includegraphics[width=0.95\linewidth,height=0.17\textheight]{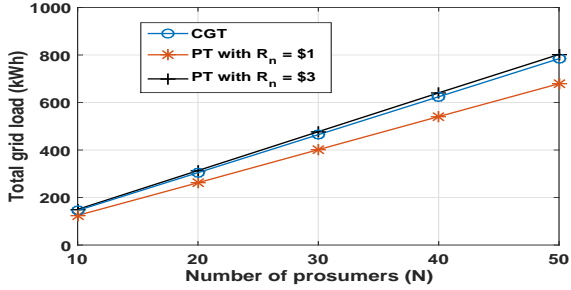}
  \caption{\label{fig:5}Total grid load for different number of prosumers}
  \end{center}\vspace{-0.3cm}
\end{figure}

\begin{figure}[t]\label{fig:6}
  \begin{center}\vspace{0.2cm}
\hspace{-0.2cm}\includegraphics[width=0.97\linewidth,height=.17\textheight]{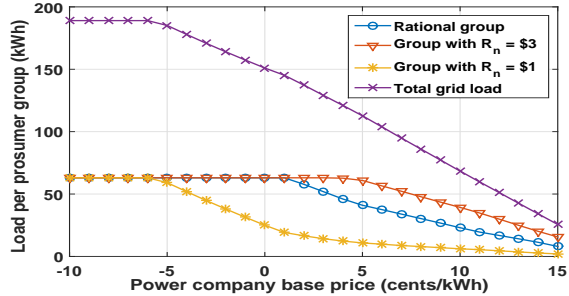}
  \caption{\label{fig:6}Load of different groups in a single grid.}
  \end{center}\vspace{-0.3cm}
\end{figure}

\indent Fig. \ref{fig:4} compares the company's profits for the scenario in which the power company accounts for prosumer irrationality to the scenario in which the power company assumes that prosumers are rational. In both scenarios, the prosumers are irrational. For a reference point below $\$1$, the  company's profits are barely affected. However, as the reference point crosses $\$1$, the company's profits start to show a clear decrease between the two scenarios. In fact, as the power company is not accounting for the prosumer's actual subjective behavior, its pricing strategy is no longer optimal. As was seen in Fig. \ref{fig:2}, this is the reference point range where the total consumption mostly differs between CGT and PT.
 The decrease in profits reaches a peak value of 15 \% at a reference point of \$2. Clearly, the power company will  experience a decrease in profits, if it neglects the subjective perception of prosumers. 

\indent Fig. \ref{fig:5} shows the total grid load energy consumption as function of the number of prosumers. The figure highlights the difference in consumption between rational prosumers and subjective prosumers with $R_n = \$1$, which increases significantly with the number of prosumers in the grid. This difference reaches $100$ kWh for $50$ prosumers. This highlights the impact of irrational behavior, which is prominent for larger grids. 

\indent Fig. \ref{fig:6} shows the energy consumption of different groups of prosumers, with different reference points, inside a single grid.
For a very small $\rho_{\textrm{base}}$, the different groups have equal consumption. As  $\rho_{\textrm{base}}$ is increased to $-5$ cents, the prosumers with  $R_n = \$1$ start to decrease their consumption at equilibrium, while the other groups' consumption remains unchanged. This is similar to what was discussed in Fig \ref{fig:2}, where prosumers with reference points close to $\$1$, exhibit risk averse behavior and thus lower energy consumption. On the other hand, rational prosumers will start decreasing their consumption at $\rho_{\textrm{base}} = 2$ cents, while risk seeking prosumers ($R_n = \$ 3$) will start decreasing their consumption at $\rho_{\textrm{base}} = 5$ cents. 

\begin{figure}[t]\label{fig:1}
\vspace{0.2cm}
  \begin{center}
\includegraphics[width=0.86\linewidth,height=.16\textheight]{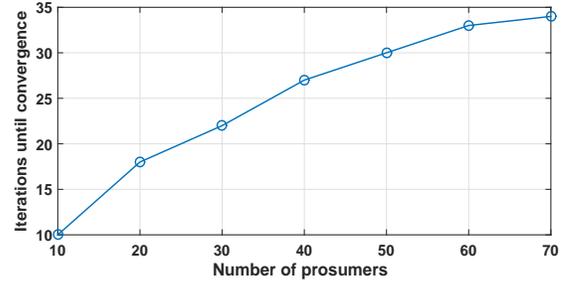}
  \caption{\label{fig:1}Number of iterations for convergence}
  \end{center}\vspace{-0.3cm}
\end{figure}

\indent Fig. \ref{fig:1} shows the number of iterations needed for the best response algorithm to converge to a followers' NE for different number of prosumers, under PT. Clearly, the best response algorithm converges, for all these cases. In addition, the number of iterations needed for convergence is reasonable, even as the number of prosumers significantly increases from $10$ to $70$.

\section{Conclusion}\label{sec:conclusion}

In this paper, we have proposed a novel framework for analyzing energy trading of prosumers with the power grid, while accounting for the uncertainty of the future price of energy. We have formulated the problem as a Stackelberg game  between the power company (leader), seeking to maximize its profits by setting its optimal pricing strategy, and multiple prosumer (followers), attempting to choose the optimal amount of energy to trade. The prosumers game was shown to have a unique pure strategy Nash equilibrium under classical game theoretic analysis. Subsequently, we have used the novel concept of utility framing from prospect theory to model the subjective behavior of prosumers when faced with the uncertainty of future energy prices. Simulation results have highlighted the impact of behavioral considerations on the overall energy trading process.

As a future avenue of research, one can extend our model to a more dynamic multi-stage game that not only utilizes further capabilities of the storage devices (e.g. load shifting over time periods), but also admits efficient algorithms for obtaining its equilibrium points. In particular, devising incentive compatible mechanisms for our model in the form of multi-stage dynamic game is an interesting future problem.

\bibliographystyle{IEEEtran}
\bibliography{thesisrefs}

\appendices

\section{Proof of Theorem \ref{thm:best-response}}\label{apx:thm_best-response}

First, we analyze the strictly concave expected utility of prosumer $n$ in (\ref{eq:utility-modified}). By taking the second derivative of (\ref{eq:utility-modified}) with respect to  $x_n$, we get: $\frac {\partial U^{\rm EUT}_{n}} {\partial^2 x_n} = - 2\alpha$, which is a strictly negative term, as $\alpha > 0$. The optimal solution is either an interior point
obtained by solving the necessary and sufficient optimality condition given by
$ \frac {\partial U^{\rm EUT}_{n}} {\partial x_{n}} = 0$, or is at one of the boundaries, in case the interior solution is not feasible. Solving the optimality solution gives a unique solution $\nonumber  x_n ^{r}=  -\frac{\theta}{2\alpha}-\frac{\bar{x}_{-n}}{2}$. $x_n ^{r}$ maximizes each prosumer's expected utility function given that it lies in the feasible range of $\mathcal{X}_n$.  

\section{Proof of Theorem \ref{thm:existence-uniqueness}}\label{apx:thm_existence-uniqueness}

First, we show that the followers' game is a \emph{concave game} with closed and convex action sets in which the utility of player $n$ is a concave function of its own action $x_n$, for any fixed actions of others $\boldsymbol{x}_{-n}$. From \eqref{eq:utility-modified}, one can see that the utility function of each prosumer $n$ is quadratic, and thus concave, in terms of its own action variable $x_n$. Moreover, the action set of each prosumer $\mathcal{X}_n$ is clearly a closed convex set. Using \cite[Theorem 1]{rosen1965existence}, we can show that the prosumers' game admits at least one pure strategy NE. For NE uniqueness, we use \cite[Theorem 2]{rosen1965existence} to show that the prosumers game is diagonally strictly concave. This means that one can find a fixed nonnegative vector $\boldsymbol{r}\ge 0$ such that for every two action profiles $\boldsymbol{x}^o, \tilde{\boldsymbol{x}}\in \mathcal{X}_1\!\times\!\cdots\!\times\!\mathcal{X}_n$, $(\tilde{\boldsymbol{x}}-\boldsymbol{x}^o)'g(\boldsymbol{x}^o,\boldsymbol{r})+(\boldsymbol{x}^o-\tilde{\boldsymbol{x}})'g(\tilde{\boldsymbol{x}},\boldsymbol{r})>0$, where $g(\boldsymbol{x},\boldsymbol{r})=(r_1\nabla_{x_1}U^{\rm EUT}_1(\boldsymbol{x}),\ldots,r_n\nabla_{x_n}U^{\rm EUT}_n(\boldsymbol{x}))'$. We let $r_j=1$ for each $j\in \mathcal{N}$. Using  \eqref{eq:utility-modified}, we have
	\begin{align}\nonumber
		g_j(\boldsymbol{x},\boldsymbol{r})=-2\alpha x_j+\theta+\alpha \bar{x}_{-j}, \ \ j\in\mathcal{N}.
	\end{align}
We let $\boldsymbol{I}$ be the identity matrix, and $\boldsymbol{J}$ be a square matrix with all entries equal to 1. Then we can write $g(\boldsymbol{x},\boldsymbol{r})=\boldsymbol{K}\boldsymbol{x}+\boldsymbol{c}$, 
where $\boldsymbol{K}:=-\alpha (\boldsymbol{I}+\boldsymbol{J})$. $\boldsymbol{K}$ is a negative definite matrix due to the positive definiteness of $\boldsymbol{I}+\boldsymbol{J}$ and the fact that $-\alpha<0$. By checking the diagonally strict concavity condition we get

{\small\begin{align}\label{eq:diagonally-strict-concave}
		&(\tilde{\boldsymbol{x}}-\boldsymbol{x}^o)'g(\tilde{\boldsymbol{x}},\boldsymbol{r})+(\boldsymbol{x}^o-\tilde{\boldsymbol{x}})'g(\boldsymbol{x}^o,\boldsymbol{r})\cr 
		&\qquad=(\tilde{\boldsymbol{x}}-\boldsymbol{x}^o)'[\boldsymbol{K}\boldsymbol{x}^o+\boldsymbol{c}]+(\boldsymbol{x}^o-\tilde{\boldsymbol{x}})'[\boldsymbol{K}\tilde{\boldsymbol{x}}+\boldsymbol{c}]\cr 
		&\qquad=-(\tilde{\boldsymbol{x}}-\boldsymbol{x}^o)'\boldsymbol{K}(\tilde{\boldsymbol{x}}-\boldsymbol{x}^o)>0,
	\end{align}} 
where the last inequality is due to the negative-definiteness of the matrix $\boldsymbol{K}$. Using \cite[Theorem 2]{rosen1965existence} the NE will be unique.  

\smallskip
\section{Auxililiary lemma for the proof of Theorem \ref{thm:convergence-rate}}\label{apx:lemm:Nikaido}
{\linespread{1.1}
\begin{lemma}\label{lemm:Nikaido}
\normalfont There exists a constant $K>0$ for which the Nikaido-Isoda function $\Psi(\boldsymbol{x},\boldsymbol{y})$ associated with the prosumers' game satisfies $\Psi(\boldsymbol{x},\boldsymbol{y})\leq K\|\boldsymbol{x}-\boldsymbol{y}\|$. Moreover, $\Psi(\boldsymbol{x},\boldsymbol{y})$ is convex in $\boldsymbol{x}$ and strongly concave in $\boldsymbol{y}$ such that 
{\small\begin{align}\label{eq:strong-concave} 
\Psi(\boldsymbol{x},\lambda \tilde{\boldsymbol{y}}+(1-\lambda)\hat{\boldsymbol{y}})&= \lambda\Psi(\boldsymbol{x}, \tilde{\boldsymbol{y}})+(1-\lambda)\Psi(\boldsymbol{x}, \hat{\boldsymbol{y}})\cr 
&+\alpha\lambda(1-\lambda)\|\hat{\boldsymbol{y}}-\tilde{\boldsymbol{y}}\|^2, \ \forall \lambda\in [0,1].\cr
\end{align}}\vspace{-0.2cm}
\end{lemma}}
\vspace{-0.7cm}

{\linespread{0.95}
\begin{proof}
For any two action profiles of the prosumers $\boldsymbol{x}=(x_1,\ldots,x_n)\in\Omega$ and $\boldsymbol{y}=(y_1,\ldots,y_n)\in\Omega$, the Nikaido-Isoda function adopted for the utility  in \eqref{eq:utility-modified} will be:

\begin{small}
\begin{align}\label{eq:nikaido}
\Psi(\boldsymbol{x},\boldsymbol{y}):&= \sum\limits_{n \in \mathcal{N}}\ [U^{\rm EUT}_n(y_n,\boldsymbol{x}_{-n},\rho_{base})-U^{\rm EUT}_n(x_n,\boldsymbol{x}_{-n},\rho_{\textrm{base}})]\cr
&=\sum\limits_{n \in \mathcal{N}} [\alpha (x_n^2-y^2_n)+(\theta+\alpha \bar{x}_{-n})(x_n-y_n)].
\end{align}\end{small}
\vspace{-0.3cm}

\noindent Using \eqref{eq:nikaido}, for any two action profiles $\boldsymbol{x},\boldsymbol{y}\in \Omega$, we have 
{\small
\begin{align}\nonumber
\Psi(\boldsymbol{x},\boldsymbol{y})&=\sum\limits_{n \in \mathcal{N}}(x_n-y_n)[\alpha (x_n+y_n)+\theta+\alpha \bar{x}_{-n}]\cr 
&\leq \sqrt{    \sum\limits_{n \in \mathcal{N}}(x_n\!-\!y_n)^2}\sqrt{\sum\limits_{n \in \mathcal{N}}[\alpha (x_n\!+\!y_n)\!+\!\theta\!+\!\alpha \bar{x}_{-n}]^2}\cr 
&=\|\boldsymbol{x}-\boldsymbol{y}\|\sqrt{\sum\limits_{n \in \mathcal{N}}[\alpha (x_n+y_n)+\theta+\alpha \bar{x}_{-n}]^2}\cr 
&\leq K \|\boldsymbol{x}-\boldsymbol{y}\|,
\end{align} \vspace{-0.35cm}}

\noindent where the first inequality is due to the Cauchy-Schwarz inequality, and $K:=\sqrt{n(\theta+\alpha(n+1)B_{\max})^2}$ is an upper bound constant for the second term of the last equality. To show the convexity of $\Psi(\boldsymbol{x},\boldsymbol{y})$ with respect to $x$, let $\boldsymbol{J}$ be the $n\times n$ matrix with all entries equal to 1. Using \eqref{eq:nikaido}, a simple calculation shows that $\nabla^2_{\boldsymbol{x}\boldsymbol{x}}\Psi(\boldsymbol{x},\boldsymbol{y})=2\alpha \boldsymbol{J}$, where $\nabla^2_{\boldsymbol{x}\boldsymbol{x}}\Psi(\boldsymbol{x},\boldsymbol{y})$ denotes the Hessian matrix of $\Psi(\boldsymbol{x},\boldsymbol{y})$ with respect to variable vector $\boldsymbol{x}$. Since $\alpha>0$ and $\boldsymbol{J}$ is a positive semi-definite matrix, this shows that $\nabla^2_{\boldsymbol{x}\boldsymbol{x}}\Psi(\boldsymbol{x},\boldsymbol{y})>0$, which implies $\Psi(\boldsymbol{x},\boldsymbol{y})$ is a convex function of $\boldsymbol{x}$. Finally using \eqref{eq:utility-modified}, one can easily check that the equality in \eqref{eq:strong-concave} holds, which shows that $\Psi(\boldsymbol{x},\boldsymbol{y})$ is strongly concave with respect to its second argument $\boldsymbol{y}$.
\end{proof}
}

\section{Proof of Theorem \ref{thm:convergence-rate}}\label{apx:thm:convergence-rate}
\vspace{-0.15cm}

We show that $\lim_{t\to \infty}\boldsymbol{x}(t)=\boldsymbol{x}^*$, from Algorithm \ref{alg:relaxing}, where $\boldsymbol{x}^*$ is a pure-strategy NE for the prosumers. To show that, we measure the distance of an action profile $\boldsymbol{x}(t)$ and its best response $\Pi_{\Omega}\big[\boldsymbol{a}+\boldsymbol{A}\boldsymbol{x}(t)\big]$ using the Nikaido-Isoda function and show that this distance decreases as $t$ becomes large. In particular, we show that at the limit, this distance equals zero which shows that the limit point is an NE of the game. 

\vspace{-0.1cm}
\begin{small}
\begin{align}\label{eq:convexity}
&\Psi(\boldsymbol{x}(t\!+\!1),\boldsymbol{x}^{r}(t\!+\!1))=\Psi\Big((1\!-\!\frac{1}{\sqrt{t}})\boldsymbol{x}(t)\!+\!\frac{\boldsymbol{x}^{r}(t)}{\sqrt{t}},\boldsymbol{x}^{r}(t\!+\!1)\Big)\cr 
&\qquad\leq(1\!-\!\frac{1}{\sqrt{t}})\Psi\Big(\boldsymbol{x}(t),\boldsymbol{x}^{r}(t)\Big)+\frac{1}{\sqrt{t}}\Psi\Big(\boldsymbol{x}^{r}(t),\boldsymbol{x}^{r}(t+1)\Big). 
\end{align}
\end{small}Using the first part of Lemma \ref{lemm:Nikaido}, we have 

\vspace{-0.1cm}
\begin{small}
\begin{align}\label{eq:two-best-distance} 
&\Psi\Big(\boldsymbol{x}^{r}(t),\boldsymbol{x}^{r}(t+1)\Big)\leq K \|\boldsymbol{x}^{r}(t)-\boldsymbol{x}^{r}(t+1)\|\cr 
&\qquad=K\|\Pi_{\Omega}\big[\boldsymbol{a}+\boldsymbol{A}\boldsymbol{x}(t)\big]-\Pi_{\Omega}\big[\boldsymbol{a}+\boldsymbol{A}\boldsymbol{x}(t+1)\big]\|\cr 
&\qquad\leq K\|[\boldsymbol{a}+\boldsymbol{A}\boldsymbol{x}(t)]-[\boldsymbol{a}+\boldsymbol{A}\boldsymbol{x}(t+1)]\|\cr 
&\qquad\leq K\|\boldsymbol{A}\|\|\boldsymbol{x}(t)-\boldsymbol{x}(t+1)\|,\cr 
&\qquad=\frac{K(n-1)}{2\sqrt{t}}\|\boldsymbol{x}(t)-\boldsymbol{x}^{r}(t)\|.
\end{align}\end{small}where the first inequality is due to the nonexpansive property of the projection operator, the second inequality uses the matrix norm inequality, and the last equality is obtained by replacing the expression for $\boldsymbol{x}(t+1)$ and noting that the induced 2-norm of matrix $\boldsymbol{A}$ equals $\frac{n-1}{2}$. Substituting \eqref{eq:two-best-distance} into \eqref{eq:convexity} we have
\begin{align}\nonumber
\Psi\Big(\boldsymbol{x}(t+1),\boldsymbol{x}^{r}(t+1)\Big)&\leq(1-\frac{1}{\sqrt{t}})\Psi\Big(\boldsymbol{x}(t),\boldsymbol{x}^{r}(t)\Big)\cr 
&\qquad+\frac{K(n-1)}{2t}\|\boldsymbol{x}(t)-\boldsymbol{x}^{r}(t)\|. 
\end{align}
Since $\Psi(\boldsymbol{x}(t),\boldsymbol{x}(t))=0$, we can write

\vspace{-0.1cm}
\begin{small}
\begin{align}\nonumber 
&\Psi\Big(\boldsymbol{x}(t+1),\boldsymbol{x}^{r}(t+1)\Big)\leq (1-\frac{1}{\sqrt{t}})\Psi(\boldsymbol{x}(t),\boldsymbol{x}^{r}(t))\cr 
&\qquad+\frac{1}{\sqrt{t}}\Psi(\boldsymbol{x}(t),\boldsymbol{x}(t))+\frac{K(n-1)}{2t}\|\boldsymbol{x}(t)-\boldsymbol{x}^{r}(t)\|\cr 
&\qquad =\Psi\left(\boldsymbol{x}(t),(1-\frac{1}{\sqrt{t}})\boldsymbol{x}^{r}(t)+\frac{1}{\sqrt{t}}\boldsymbol{x}(t)\right)\cr
&\qquad-\alpha(1\!-\!\frac{1}{\sqrt{t}})\frac{1}{\sqrt{t}}\|\boldsymbol{x}(t)-\boldsymbol{x}^{r}(t)\|^2\!+\!\frac{K(n-1)}{2t}\|\boldsymbol{x}(t)-\boldsymbol{x}^{r}(t)\|\cr 
&\qquad\leq \Psi(\boldsymbol{x}(t),\boldsymbol{x}^{r}(t))-\frac{\Psi^2(\boldsymbol{x}(t),\boldsymbol{x}^{r}(t))}{\frac{2K^2}{\alpha}\sqrt{t}}+\frac{K(n-1)D}{2t}, 
\end{align}\end{small}where the first equality is due to Lemma \ref{lemm:Nikaido}, and the last inequality is due to first part of Lemma \ref{lemm:Nikaido} and the fact that $\boldsymbol{x}^{r}(t)$ maximizes $\Psi(\boldsymbol{x}(t),\cdot)$. Multiplying both sides of the above inequality by $\frac{\alpha}{2K^2\sqrt{t}}$ and defining $c:=\frac{\alpha(n-1)D}{4K}$ and $a_t:=\frac{\alpha}{2K^2\sqrt{t}}\Psi(\boldsymbol{x}(t),\boldsymbol{x}^{r}(t))$, we get  
\begin{align}\label{eq:recursive}
a_{t+1}\leq a_t-a^2_t+\frac{c}{t\sqrt{t}}.
\end{align}

Our goal is to show that $a_t<\sqrt{2c}\times t^{-\frac{3}{4}}$ for all $t\ge \frac{100}{c^2}$, in which case by definition of $a_t$ we obtain $\Psi(\boldsymbol{x}(t),\boldsymbol{x}^{r}(t))=O(t^{-\frac{1}{4}})$. This not only shows that $\lim_{t\to \infty}\Psi(\boldsymbol{x}(t),\boldsymbol{x}^{r}(t))=0$, implying that $\{\boldsymbol{x}(t)\}$ converges to a pure strategy NE of the prosumers game (note that $\Psi(\boldsymbol{x},\boldsymbol{x}^{r})=0$ if and only if $\boldsymbol{x}$ is a NE), but it also shows that after $t$ steps, the action profile of the prosumers $\boldsymbol{x}(t)$ is an $\epsilon$-NE of the game where $\epsilon=O(t^{-\frac{1}{4}})$ (this is due to $\Psi(\boldsymbol{x}(t),\boldsymbol{x}^{r}(t))=O(t^{-\frac{1}{4}})$ implies $U_n(\boldsymbol{x}^{r}_n(t),\boldsymbol{x}_{-n}(t),\rho_{\textrm{base}})-U_n(x_n(t),\boldsymbol{x}_{-n}(t),\rho_{\textrm{base}})=O(t^{-\frac{1}{4}})$ for all $n\in\mathcal{N}$, meaning that given the action profile $\boldsymbol{x}(t)$, no prosumer can increase its utility by more that $O(t^{-\frac{1}{4}})$ by playing its best response). 

We complete the proof using induction on $t$ to show that $a_t<\sqrt{2c}\times t^{-\frac{3}{4}}$. Assume that this relation is true for $t$. Then
{\small
\begin{align}\nonumber
a_{t+1}&\leq a_t-a^2_t+\frac{c}{t\sqrt{t}}\cr 
&\leq \sqrt{2c}t^{-\frac{3}{4}}-2ct^{-\frac{3}{2}}+\frac{c}{t\sqrt{t}}\cr 
&=\sqrt{2c}t^{-\frac{3}{4}}-ct^{-\frac{3}{2}}.
\end{align}}
Let $f(z):[1,\infty)\to\mathbb{R}$ be a function defined by $f(z)=\sqrt{2c}z^{-\frac{3}{4}}-\sqrt{2c}(z+1)^{-\frac{3}{4}}-cz^{-\frac{3}{2}}$. We only need to show that $f(z)<0,$ for $t\ge \frac{100}{c^2}$. By writing the Taylor expansion of the first two terms of $f(z)$ for $z\ge 1$, we have $f(z)\leq 7\sqrt{2c}z^{-\frac{7}{4}}-cz^{-\frac{3}{2}}$, which is less than $0$ for $t\ge \frac{100}{c^2}$. This completes the induction and shows that $a_t=O(t^{-\frac{3}{4}})$.

\section{Proof of Theorem \ref{thm:PT_Cases}}\label{apx:thm:PT_Cases}
We first  find the conditions under which the expected utility function is uniform over each prosumer's action space. We then find the additional condition to ensure that the function is strictly concave.

\smallskip
\indent{\bf Case 1:} To have  $R_n < \rho_{\textrm{min}} c + d$, for all of prosumer $n$'s actions, we first rewrite the inequality in terms of $x_n$:
\begin{align}\nonumber
-\alpha x_n^2 + \left( \rho_{\textrm{min}} - \rho_b - \alpha \bar{x}_{-n} \right) x_n + \left( \rho_{\textrm{min}}K_n - R_n \right) < 0,
\end{align}
 where $k_n = W_n + Q_n - L_n$. By analyzing the second order polynomial, and its roots, $r_1$ and $r_2$, we get the condition for case $1$.
Under such a condition, the expected utility function of prosumer $n$ under PT, simplifies  to $U^{\textrm{PT}}_{n}(x_n,\bar{x}_{-n},\rho_{\textrm{base}}) = U^{\textrm{EUT}}_{n}(x_n,\bar{x}_{-n},\rho_{\textrm{base}}) - R_n$. This is clearly a concave function, given that,  $U^{\textrm{EUT}}_{n}(x_n,\bar{x}_{-n},\rho_{\textrm{base}})$ has been shown to be concave, and $R_n$ is a constant.

\smallskip
\indent{\bf Case 2:} In order to have  $R_n > \rho_{\textrm{max}} c + d$, for all of prosumer $n$'s actions, we follow a similar approach in order to find the condition for case $2$. Under this condition, the expected utility function under PT simplifies  to $U^{\textrm{PT}}_{n}(x_n,\bar{x}_{-n},\rho_{base}) = \lambda (U^{\textrm{EUT}}_{n}(x_n,\bar{x}_{-n},\rho_{\rm base}) - R_n)$, which is also strictly concave, given that $\lambda$ is strictly positive.

\smallskip
\indent{\bf Case 3:} To have  $ \rho_{\textrm{min}} c + d   < R_n < \rho_{\textrm{max}} c + d$, for all of prosumer $n$'s actions, we follow a similar approach in order to find the condition for case $3$. We next analyze the concavity of the expected utility function, given in the second line in (\ref{eq:PT}). The second derivative is given by:

\begin{small}
\begin{align}\nonumber
\frac {\partial U_{n,PT}} {\partial^2 x_n} = \frac{a_1}{m_1}x_n - \frac{a_1m_1-bm_1}{m_1^2} - \frac{(\lambda-1)a_2^2}{(Q_n - L_n + W_n + x_n)^3},
\end{align}\end{small}

where $a_2 = (R_n - (k_n \rho_{\textrm{base}}) + \alpha(L_n^2  + Q_n^2  +  w_n^2) - 2 L_n Q_n \alpha + L_n \alpha\bar{x}_{-n} - Q_n \alpha \bar{x}_{-n} - 2 L_n \alpha w_n + 2 Q_n\alpha w_n$. Note that $-\frac{(\lambda-1)a_2^2}{(Q_n - L_n + W_n + x_n)^3}$ is negative for all $x_n$. Next, we find the range of $x_n$ for which $\frac{a_1}{m_1}x_n - \frac{a_1m_1-b_1m_1}{m_1^2}$ is negative as well. Given that $\frac{a_1}{m_1}$ is negative, the utility function is thus guaranteed to be concave for $x_n > \frac{a_1m_1-b_1m_1}{a_1m_1}$.

\end{document}